\documentclass[10pt]{article}

\usepackage[a4paper,margin=2.4cm]{geometry}
\usepackage[T1]{fontenc}
\usepackage{lmodern}
\usepackage{microtype}
\usepackage{amsmath,amssymb,amsthm,mathtools}
\usepackage{dsfont}
\usepackage{xcolor}
\usepackage[colorlinks=true,linkcolor=black,citecolor=black,urlcolor=black]{hyperref}

\newtheorem{theorem}{Theorem}[section]
\newtheorem{proposition}[theorem]{Proposition}
\theoremstyle{remark}
\newtheorem{remark}[theorem]{Remark}
\newcommand{\Tr}{\operatorname{Tr}}
\newcommand{\id}{\mathds{1}}
\newcommand{\Hmin}{H_{\min}}
\newcommand{\SN}{\operatorname{SN}}
\newcommand{\norm}[1]{\left\lVert #1\right\rVert}
\newcommand{\abs}[1]{\left\lvert #1\right\rvert}

\title{Sharp continuity of quantum conditional entropy}
\author{%
 Mario Berta$^{1}$, Pablo Costa Rico$^{1}$, Gereon Kossmann$^{1}$,\\
 Ludovico Lami$^{2}$, and Julius A.\ Zeiss$^{1}$\\[0.8em]
 \parbox{0.92\textwidth}{\centering\small
 $^{1}$Institute for Quantum Information, RWTH Aachen University,
 Aachen, Germany\\
 $^{2}$Scuola Normale Superiore, Pisa, Italy}}
\date{}
\hypersetup{
 pdftitle={Sharp continuity of quantum conditional entropy},
 pdfauthor={Mario Berta, Pablo Costa Rico, Gereon Kossmann,
 Ludovico Lami, Julius A. Zeiss}
}

\begin{document}
\maketitle

\begin{abstract}
We prove the sharp uniform continuity bound for quantum conditional
entropy.  If two bipartite states are at trace distance at most $\delta$
and $d=\dim A$, the optimal dimension-only modulus of continuity is
$h_2(\delta)+\delta\log(d^2-1)$ up to $\delta=1-d^{-2}$ and $2\log d$
thereafter, where $h_2$ denotes the binary entropy.  When $\dim B\ge d$,
this bound is tight for every $\delta\in[0,1]$.  The key proof idea was
developed with the assistance of
ChatGPT 5.6 Sol,
building on and adapting the tight classical proof of Alhejji \& Smith
[IEEE ISIT (2020)], which follows a conceptually different approach.
\end{abstract}

\section*{Introduction}
The first uniform continuity estimate for quantum conditional entropy whose modulus depends only on the dimension of the conditioned system $A$, and not on that of the conditioning system $B$, was proved by Alicki and Fannes
\cite{AlickiFannes2004} and later strengthened by Winter
\cite{Winter2016}. Alhejji and Smith then determined the sharp
classical modulus of continuity \cite{AlhejjiSmith2020}. Wilde proved the
corresponding bound when the conditioned system is quantum and the
conditioning system is classical, and conjectured the fully quantum
modulus of continuity with $\log(d^2-1)$, where $d=\dim A$ \cite{Wilde2020}.  Berta,
Lami, and Tomamichel \cite{BertaLamiTomamichel2025} and, independently,
Audenaert, Bergh, Datta, Jabbour, Capel, and Gondolf
\cite{AudenaertEtAl2025} established the conjecture when the
two states have the same $B$-marginal and left the unrestricted case
open.
Theorem~\ref{thm:main} settles this problem.  Appendix~\ref{app:classical}
gives a global classical proof without the conditional-fibre decomposition
used by Alhejji and Smith, while Appendix~\ref{app:winter} compares the
argument with Winter's proof.\\

\textit{We would like to highlight again that this manuscript and its content were 
developed with the assistance of ChatGPT~5.6 Sol, a generative AI tool. 
The manuscript aims to present a solution to an open problem in quantum 
information theory, which we believe is of interest to a significant part 
of our community. Moreover, it may stimulate the recently initiated 
discussion about the use of AI-generated scientific content and the 
appropriate format for its communication.}

\section{Main result}

All systems are finite dimensional and all logarithms are natural.  For a
state $\omega_{AB}$, write
\[
 H(\omega_{AB})=-\Tr\,\omega_{AB}\log\omega_{AB},\qquad
 H(A|B)_\omega=H(\omega_{AB})-H(\omega_B),
\]
\[
 T(\rho_{AB},\sigma_{AB})
 =\frac12\norm{\rho_{AB}-\sigma_{AB}}_1,\qquad
 h_2(t)=-t\log t-(1-t)\log(1-t).
\]
Here $0\log0\coloneqq 0$, and
\[
 D(\omega_{AB}\|\eta_{AB})
 \coloneqq \Tr\,\omega_{AB}(\log\omega_{AB}-\log\eta_{AB})
\]
denotes quantum relative entropy, with the usual support convention.

\begin{theorem}[Sharp uniform continuity bound]\label{thm:main}
Let $\rho_{AB}$ and $\sigma_{AB}$ be states, let $d=\dim A\ge2$, and suppose
\[
 T(\rho_{AB},\sigma_{AB})\le\delta\le1.
\]
Then
\[
 \boxed{\;
 \abs{H(A|B)_\rho-H(A|B)_\sigma}
 \le
 \begin{cases}
  h_2(\delta)+\delta\log(d^2-1),
      &0\le\delta\le1-d^{-2},\\[1mm]
  2\log d,
      &1-d^{-2}\le\delta\le1.
 \end{cases}\;}
\]
If $\dim B\ge d$, the right-hand side is optimal for every
$\delta\in[0,1]$.
\end{theorem}

\begin{proof}
It suffices to prove the bound without the absolute value; the reverse
inequality follows by interchanging $\rho_{AB}$ and $\sigma_{AB}$.  The
case $\delta=0$ is immediate, while the second branch follows from
$-\log d\le H(A|B)\le\log d$.  It remains to prove the first branch, so suppose
$0<\delta\le1-d^{-2}$.  Replacing $\rho_{AB}$ and $\sigma_{AB}$ by
$\rho_{\varepsilon,AB}=(1-\varepsilon)\rho_{AB}
+\varepsilon\id_{AB}/\dim(A\otimes B)$
and $\sigma_{\varepsilon,AB}=(1-\varepsilon)\sigma_{AB}+
\varepsilon\id_{AB}/\dim(A\otimes B)$ preserves the trace-distance
hypothesis; continuity as $\varepsilon\downarrow0$ therefore allows us to
assume that both states are strictly positive.

The pinching inequality in any orthonormal basis
$\{\lvert i\rangle\}_{i=1}^d$ of $A$ gives, for every $X_{AB}\ge0$,
\begin{equation}\label{eq:domination}
 X_{AB}\le
 d\sum_{i=1}^d
 (\lvert i\rangle\!\langle i\rvert\otimes\id_B)
 X_{AB}
 (\lvert i\rangle\!\langle i\rvert\otimes\id_B)
 \le d\,\id_A\otimes X_B .
\end{equation}
Define the states
\begin{equation}\label{eq:states}
 \widehat{\sigma}_{AB}
 \coloneqq \frac{d\,\id_A\otimes\sigma_B-\sigma_{AB}}{d^2-1},
 \qquad
 \tau_{AB}\coloneqq (1-\delta)\sigma_{AB}
              +\delta\widehat{\sigma}_{AB}.
\end{equation}
Indeed, \eqref{eq:domination} shows that $\widehat{\sigma}_{AB}\ge0$;
its trace is one and $\widehat{\sigma}_B=\sigma_B$, so
$\tau_B=\sigma_B$.  Moreover,
\[
 \tau_{AB}
 =\left(1-\frac{d^2\delta}{d^2-1}\right)\sigma_{AB}
  +\frac{d\delta}{d^2-1}\,\id_A\otimes\sigma_B .
\]
Since the first coefficient is nonnegative, we  obtain the lower bound
\[
    \tau_{AB}\ge \frac{d\delta}{d^2-1}\,\id_A\otimes\sigma_B.
\]
On the other hand, \eqref{eq:domination} yields
\[
    \tau_{AB}\le \left(1-\frac{d^2\delta}{d^2-1}\right)d \, \id_A\otimes \sigma_B+\frac{d\delta}{d^2-1}\,\id_A\otimes\sigma_B
 = d(1-\delta)\,\id_A\otimes\sigma_B\, ,
\]
and we obtain
\begin{equation}\label{eq:sandwich}
 \frac{d\delta}{d^2-1}\,\id_A\otimes\sigma_B
 \le\tau_{AB}
 \le d(1-\delta)\,\id_A\otimes\sigma_B .
\end{equation}

Put
\[
 G_{AB}\coloneqq -\log\tau_{AB}+\id_A\otimes\log\sigma_B .
\]
Direct expansion gives
\begin{align*}
 H(A|B)_\rho-H(A|B)_\sigma
 &=
 \Tr[(\rho_{AB}-\sigma_{AB})G_{AB}]
 +D(\sigma_{AB}\|\tau_{AB})-D(\rho_{AB}\|\tau_{AB})
 +D(\rho_B\|\sigma_B)\\
 &\le \Tr[(\rho_{AB}-\sigma_{AB})G_{AB}]
       +D(\sigma_{AB}\|\tau_{AB}),
\end{align*}
where the inequality follows from the data-processing inequality for relative entropy under the partial trace over $A$, using $\tau_B=\sigma_B$.

By operator monotonicity of the logarithm, \eqref{eq:sandwich} implies
\[
\log\left(\frac{d\delta}{d^2-1}\right)\id_{AB}+\id_A \otimes \log \sigma_B\le \log \tau_{AB}\le \log(d(1-\delta))\id_{AB}+\id_A \otimes \log \sigma_B \, .
\]
Multiplying by $-1$ and adding $\id_A\otimes \log \sigma_B$, we obtain
\[
 -\log\!\bigl(d(1-\delta)\bigr)\id_{AB}
 \le G_{AB}\le
 -\log\!\left(\frac{d\delta}{d^2-1}\right)\id_{AB}.
\]
Put $t=T(\rho_{AB},\sigma_{AB})$. If $t=0$, the claim is immediate, so suppose $t>0$. The Jordan decomposition of the
traceless operator $\Delta_{AB}\coloneqq \rho_{AB}-\sigma_{AB}$ is given by $\Delta_{AB}=t(\Delta_{+,AB}-\Delta_{-,AB})$, where $\Delta_{\pm,AB}\geq 0$, $\Delta_{+,AB}\Delta_{-,AB}=0$ and $\Tr[\Delta_{+,AB}]=\Tr[\Delta_{-,AB}]=1$. Then,
\begin{align*}
\Tr[\Delta_{AB}G_{AB}]
&=t\Tr[\Delta_{+,AB}G_{AB}]-t\Tr[\Delta_{-,AB}G_{AB}]\\
&\le -t\log\left(\frac{d\delta}{d^2-1}\right)\Tr[\Delta_{+,AB}]
+t\log(d(1-\delta))\Tr[\Delta_{-,AB}],
\end{align*}
which gives
\[
 \Tr[(\rho_{AB}-\sigma_{AB})G_{AB}]
 \le t\log\frac{(d^2-1)(1-\delta)}{\delta}
 \le\delta\log\frac{(d^2-1)(1-\delta)}{\delta},
\]
where the last logarithm is nonnegative  in the assumed range.  Finally,
$\tau_{AB}\ge(1-\delta)\sigma_{AB}$ implies
$D(\sigma_{AB}\|\tau_{AB})\le-\log(1-\delta)$.  Consequently,
\[
 H(A|B)_\rho-H(A|B)_\sigma
 \le
 \delta\log\frac{(d^2-1)(1-\delta)}{\delta}
 -\log(1-\delta)
 =h_2(\delta)+\delta\log(d^2-1).
\]

For sharpness, let $B_0\subseteq B$ have dimension $d$, let
$\Phi_{AB}$ be a maximally entangled pure state on $A\otimes B_0$, and set
$s=\min\{\delta,1-d^{-2}\}$.  The states
\[
 \sigma_{AB}=\Phi_{AB},\qquad
 \rho_{AB}=(1-s)\Phi_{AB}
      +\frac{s}{d^2-1}(\id_{A\otimes B_0}-\Phi_{AB})
\]
satisfy $T(\rho_{AB},\sigma_{AB})=s\le\delta$, have the same
$B$-marginal, and obey
\[
 H(A|B)_\rho-H(A|B)_\sigma
 =H(\rho_{AB})=h_2(s)+s\log(d^2-1).
\]
This is the first branch when $\delta\le1-d^{-2}$ and equals
$2\log d$ when $\delta\ge1-d^{-2}$.
\end{proof}

\begin{remark}\label{rem:weyl}
For a Weyl unitary error basis
$\{W_{jk}\}_{j,k=0}^{d-1}$ on $A$, we have
\[
 \widehat{\sigma}_{AB}
 =\frac1{d^2-1}\sum_{(j,k)\ne(0,0)}
 (W_{jk}\otimes\id_B)\sigma_{AB}
 (W_{jk}^{\dagger}\otimes\id_B).
\]
Thus the comparison state in the proof is obtained by hedging the anchor
$\sigma_{AB}$ toward the uniform mixture of its $d^2-1$ nonidentity Weyl-error sectors. 
\end{remark}

\section{Refinements}
\label{sec:refinements}

Let $m=\dim B$, $r=\min\{d,m\}$, and
$\pi_A=\id_A/d$.  For $K>1$, define the nondecreasing function
\begin{equation}\label{eq:gK}
 g_K(\delta)\coloneqq 
 \begin{cases}
  h_2(\delta)+\delta\log(K-1),
      &0\le\delta\le1-K^{-1},\\[1mm]
  \log K,
      &1-K^{-1}\le\delta\le1,
 \end{cases}
 \qquad g_1\equiv0.
\end{equation}
For a state $\omega_{AB}$, define its conditional min-entropy relative to
its own marginal $\omega_B$ by
\[
 \Hmin(A|B)_{\omega\mid\omega}
 \coloneqq -\log\inf\{\lambda>0:
       \omega_{AB}\le\lambda\,\id_A\otimes\omega_B\},
\]
and set
\[
 \kappa_\omega
 \coloneqq d\,\exp\!\left[-\Hmin(A|B)_{\omega\mid\omega}\right].
\]
The mixed-state Schmidt number $\SN(\omega_{AB})$ is the least integer
$s$ for which $\omega_{AB}$ has a convex decomposition into pure states
of Schmidt rank at most $s$ \cite{TerhalHorodecki2000}.

\begin{proposition}[Hierarchy and exact Schmidt-number interpolation]
\label{prop:meta}
Let $\rho_{AB}$ and $\sigma_{AB}$ be states with
$T(\rho_{AB},\sigma_{AB})\le\delta\le1$.  Then
\begin{equation}\label{eq:state-dependent}
 H(A|B)_\rho-H(A|B)_\sigma
 \le
 \min\left\{
  g_{\kappa_\sigma}(\delta),\,
  \log d-H(A|B)_\sigma
 \right\},
\end{equation}
and
\begin{equation}\label{eq:hierarchy}
 1\le\kappa_\sigma
 =d\,\exp\!\left[-\Hmin(A|B)_{\sigma\mid\sigma}\right]
 \le d\,\SN(\sigma_{AB})
 \le dr
 \le d^2.
\end{equation}
Consequently,
\begin{equation}\label{eq:modulus-hierarchy}
 H(A|B)_\rho-H(A|B)_\sigma
 \le g_{\kappa_\sigma}(\delta)
 \le g_{d\,\SN(\sigma_{AB})}(\delta)
 \le g_{dr}(\delta)
 \le g_{d^2}(\delta).
\end{equation}
Moreover, for every integer $1\le s\le r$,
\begin{equation}\label{eq:exact-SN}
 \sup_{\substack{T(\rho_{AB},\sigma_{AB})\le\delta\\
                  \SN(\rho_{AB}),\,\SN(\sigma_{AB})\le s}}
 \abs{H(A|B)_\rho-H(A|B)_\sigma}
 =g_{ds}(\delta).
\end{equation}
Thus both the classical and the separable cases have the sharp modulus of continuity
$g_d$; Schmidt number at most $s$ gives the interpolation $g_{ds}$; and
the unrestricted quantum case has the sharp modulus of continuity $g_{dr}$.  In
particular, if $m\ge d$, then $g_{dr}=g_{d^2}$, which reaches its plateau value $2\log d$ at $\delta=1-d^{-2}$.
\end{proposition}

\begin{proof}
First suppose that
$\sigma_{AB}\le K\pi_A\otimes\sigma_B$ with $K>1$, and put
\[
 \widehat{\sigma}^{(K)}_{AB}
 \coloneqq \frac{K\pi_A\otimes\sigma_B-\sigma_{AB}}{K-1},
 \qquad
 \alpha\coloneqq \min\{\delta,1-K^{-1}\},
 \qquad
 \tau^{(K)}_{AB}\coloneqq (1-\alpha)\sigma_{AB}
                 +\alpha\widehat{\sigma}^{(K)}_{AB}.
\]
These are states with the same $B$-marginal, and
\[
 \frac{K\alpha}{K-1}\,\pi_A\otimes\sigma_B
 \le\tau^{(K)}_{AB}
 \le K(1-\alpha)\,\pi_A\otimes\sigma_B.
\]
The score identity and data-processing step in the proof of
Theorem~\ref{thm:main} therefore give
\[
 H(A|B)_\rho-H(A|B)_\sigma
 \le
 \delta\log\frac{(K-1)(1-\alpha)}{\alpha}
 -\log(1-\alpha)
 =g_K(\delta),
\]
with the evident limiting interpretation when $\alpha=0$.  The case
$K=1$ follows from $\sigma_{AB}=\pi_A\otimes\sigma_B$.  For singular
states, regularise both $\rho_{AB}$ and $\sigma_{AB}$ toward
$\pi_A\otimes\pi_B$ with $\pi_B=\id_B/m$ and let $\varepsilon\downarrow0$; domination is
preserved since
\[
 K\pi_A\otimes\sigma_{\varepsilon,B}-\sigma_{\varepsilon,AB}
 =(1-\varepsilon)(K\pi_A\otimes\sigma_B-\sigma_{AB})
  +\varepsilon(K-1)\pi_A\otimes\pi_B\ge0.
\]

Taking $K=\kappa_\sigma$ proves the first term in
\eqref{eq:state-dependent}; the second follows from
\[
 H(A|B)_\rho-H(A|B)_\sigma
 \le\log d-H(A|B)_\sigma.
\]
If a pure state $\psi_{AB}$ has Schmidt rank $s$, its Schmidt
decomposition gives
$\psi_{AB}\le s\,\id_A\otimes\psi_B$.
Convex decomposition hence yields
\[
 \omega_{AB}\le\SN(\omega_{AB})\,\id_A\otimes\omega_B
 =d\,\SN(\omega_{AB})\,\pi_A\otimes\omega_B,
\]
which proves \eqref{eq:hierarchy}.  Since $g_K(\delta)$ is nondecreasing
in $K$, applying the one-sided estimate in both directions proves the
upper bound in \eqref{eq:exact-SN}.

For sharpness, choose $s$-dimensional subspaces
$A_0\subseteq A$ and $B_0\subseteq B$, let $\Phi_{s,AB}$ be a maximally
entangled state on $A_0\otimes B_0$, let $P_{B_0}$ be the orthogonal
projector onto $B_0$, put
$P_{AB}=\id_A\otimes P_{B_0}$ and $M=ds$, and set
\[
 t\coloneqq \min\{\delta,1-M^{-1}\},\qquad
 \rho_{t,AB}=(1-t)\Phi_{s,AB}
              +\frac{t}{M-1}(P_{AB}-\Phi_{s,AB}),
 \qquad
 \sigma_{AB}=\Phi_{s,AB}.
\]
Both states are supported on $A\otimes B_0$, so they have Schmidt
number at most $s$. They have the same $B$-marginal and satisfy
$T(\rho_{t,AB},\sigma_{AB})=t$.  Consequently,
\[
 H(A|B)_{\rho_t}-H(A|B)_\sigma
 =h_2(t)+t\log(M-1)=g_M(\delta).
\]
For $s=1$ this construction is diagonal in a product basis, so it also
proves sharpness in the classical case.
\end{proof}

\section{Discussion}

We have determined the sharp continuity modulus of quantum
conditional entropy. The proof is elementary and rests on a single
construction, the canonical complement $\widehat{\sigma}_{AB}$ of
Remark~\ref{rem:weyl}. The same mechanism
gives a fixed-marginal conditional-min-entropy refinement and the exact
Schmidt-number hierarchy connecting the classical, separable, and fully
quantum regimes.  

A natural application of the continuity bound for the conditional entropy is to the continuity of the mutual information of a quantum channel. Given quantum channels $\mathcal{N},\mathcal{M}$ from $A$ to $B$, for an input state $\rho_A$, choose a purification $\psi_{RA}$ with $R\simeq A$. The mutual information of the channel $\mathcal{N}$ at $\rho_A$ is defined as
\begin{equation*}
I(\rho,\mathcal{N})
:= I(R:B)_{(\mathrm{id}_R\otimes\mathcal{N})(\psi_{RA})},
\end{equation*}
and the channel mutual information is obtained by optimizing over all input states,
\begin{equation}\label{eq:def_mutual_information}
I(\mathcal{N})
:= \sup_{\rho_A} I(\rho,\mathcal{N}).
\end{equation}
The channel mutual information coincides with the entanglement-assisted classical capacity, $
C_E(\mathcal{N})=I(\mathcal{N})$,
and the following continuity bound for $I(\mathcal{N})$ immediately yields the corresponding continuity bound for the entanglement-assisted classical capacity.

\begin{proposition}
    Let $2\leq d \leq m$ and $\delta\in[0,1]$. If $\mathcal{N},\mathcal{M}$ are quantum channels from $A$ to $B$ satisfying $\Vert \mathcal{N}-\mathcal{M}\Vert_\diamond \leq 2 \delta$, then for every input state $\rho_A$,
    \begin{equation*}
        \vert I(\rho,\mathcal{N})-I(\rho,\mathcal{M})\vert \leq g_{d^2}(\delta)\, ,
    \end{equation*}
    where $g_{d^2}(\delta)$ is given by  \eqref{eq:gK}. Consequently,
    \begin{equation}
        \vert I(\mathcal{N})-I(\mathcal{M})\vert \leq g_{d^2}(\delta)\, .
    \end{equation}
\end{proposition}
\begin{proof}
    For fixed $\rho_A$ with purification $\psi_{RA}$, let
    \begin{equation*}
        \omega_{RB}=(\mathrm{id}_R\otimes \mathcal{N})(\psi_{RA})\, , \quad \tau_{RB}=(\mathrm{id}_R\otimes \mathcal{M})(\psi_{RA})\, ,
    \end{equation*}
    which satisfy
    \begin{equation*}
        \frac{1}{2}\Vert \omega_{RB}-\tau_{RB}\Vert_1\leq \frac{1}{2}\Vert \mathcal{N}-\mathcal{M}\Vert_\diamond \leq \delta
    \end{equation*}
    by definition of the diamond distance. In addition, $\omega_R=\tau_R=\psi_R$. Hence, $H(R)_{\omega}=H(R)_{\tau}=H(\rho)$  and as a consequence
    \begin{equation*}
        \vert I(\rho,\mathcal{N})- I(\rho,\mathcal{M})\vert = \vert I(R:B)_{\omega}-I(R:B)_{\tau} \vert =\vert H(R\vert B)_{\tau}-H(R\vert B)_{\omega} \vert \, .
    \end{equation*}
    The result then follows by a direct application of Theorem \ref{thm:main}. The second part follows directly from the inequality
    \begin{equation*}
        \left|\sup_x f(x)-\sup_x g(x)\right|\leq \sup_x\vert f(x)-g(x)\vert\, .
    \end{equation*}
\end{proof}

\begin{remark}
    The bounds stated in the previous result are also sharp. For $A\simeq B \simeq \mathbb{C}^d$, consider for $0 \leq t \leq 1-d^{-2}$ the depolarizing channel
    \begin{equation*}
        \mathcal{D}_t(X)=(1-p_t)X+p_t \frac{\Tr[X]}{d}\id_{d} \, , \quad p_t=\frac{d^2}{d^2-1}t \, \in [0,1].
    \end{equation*}
    If we denote the identity channel by $\mathrm{id}_{d}$, then $\Vert \mathrm{id}_{d}-\mathcal{D}_t\Vert_\diamond=2t$. For the depolarizing channel, the supremum in \eqref{eq:def_mutual_information} is attained at the maximally mixed state $\pi_A$, by unitary covariance and concavity of $I(\rho,\mathcal{D}_t)$ in $\rho$, yielding
    \begin{equation*}
        I(\mathcal{D}_t)=I(\pi_A,\mathcal{D}_t)=I(R:B)_{J_t}=H(R)_{J_t}+H(B)_{J_t}-H(J_t)=2\log d-h_2(t)-t\log(d^2-1)\, ,
    \end{equation*}
    where $J_t=(\mathrm{id}_d\otimes \mathcal{D}_t)(\Phi_{RA})$ denotes the normalized Choi state of $\mathcal{D}_t$ and $\Phi_{RA}$ the maximally entangled state. Therefore,
    \begin{equation*}
        I(\mathrm{id}_{d})-I(\mathcal{D}_t)=h_2(t)+t\log(d^2-1)
    \end{equation*}
    shows that equality is attained by choosing $t=\min\{\delta,1-d^{-2} \}$.
    For $m>d$, choose an isometry $V:\mathbb{C}^d\to B$ and use the
    channels $X\mapsto VXV^\dagger$ and $X\mapsto V\mathcal{D}_t(X)V^\dagger$.
    The common output isometry preserves the diamond distance and channel
    mutual information, so sharpness holds for every $m\ge d$.
\end{remark}

  Further directions include bounds parameterized by
fidelity or purified distance, energy-constrained infinite-dimensional
versions, and a formulation in terms of quantum couplings that may
connect the classical and quantum arguments.

\section*{Acknowledgements}

Over the years, we have discussed the problem with many people, and we
thank all of them.  Most recently, Nilanjana Datta presented the question
in the open problem session at the workshop on ``Additivity Problems in Quantum and Classical Information
Theory''. We gratefully acknowledge
the hospitality of the Banff International Research Station for
Mathematical Innovation and Discovery in Banff, Alberta, Canada, during
the workshop ``Additivity Problems in Quantum and Classical Information
Theory [26w5621]'', July 12--17, 2026.
MB, PCR, GK, and JZ acknowledge support from the European
Research Council (ERC Grant Agreement No.\ 948139) and from the Excellence Cluster --
Matter and Light for Quantum Computing (ML4Q-2). LL acknowledges financial support from the European Union (ERC StG ETQO, Grant Agreement no.\ 101165230). 

\appendix

\section{The classical case without conditional fibres}
\label{app:classical}

For probability distributions, $T$ denotes total variation distance and
$H(X|Y)_P=H(P_{XY})-H(P_Y)$.  Alhejji and Smith proved the following sharp
bound using a conditional-fibre decomposition
\cite{AlhejjiSmith2020}.  We give a global proof without that decomposition.

\begin{proposition}[Sharp classical modulus of continuity]
\label{prop:classical}
Let $P_{XY}$ and $Q_{XY}$ be distributions with
$\lvert\mathcal X\rvert=d\ge2$ and $T(P,Q)\le\delta\le1$.  Then
\[
 \abs{H(X|Y)_P-H(X|Y)_Q}\le g_d(\delta),
\]
and the bound is optimal for every $\delta$.
\end{proposition}

\begin{proof}
The case $\delta=0$ is immediate, so assume $\delta>0$.
It is enough to prove the signed bound;the case of zero probabilities
follows by continuity.  Put
\[
 U_Q(x,y)\coloneqq \frac{Q_Y(y)}d,\qquad
 \widehat Q\coloneqq \frac{dU_Q-Q}{d-1},\qquad
 \alpha\coloneqq \min\{\delta,1-d^{-1}\},
 \qquad R\coloneqq (1-\alpha)Q+\alpha\widehat Q.
\]
Since $Q\le dU_Q$, the distribution $\widehat Q$ is nonnegative and has
marginal $Q_Y$.  Moreover,
\[
 \frac{d\alpha}{d-1}U_Q\le R\le d(1-\alpha)U_Q.
\]
For $g(x,y)\coloneqq -\log R(x,y)+\log Q_Y(y)$, the same relative-entropy
identity and marginal data processing as in Theorem~\ref{thm:main} give
\[
 H(X|Y)_P-H(X|Y)_Q
 \le\sum_{x,y}(P-Q)(x,y)g(x,y)+D(Q\|R).
\]
The pointwise sandwich makes the logarithm below nonnegative; together
with $T(P,Q)\le\delta$, total-variation duality, and
$R\ge(1-\alpha)Q$, it yields
\[
 H(X|Y)_P-H(X|Y)_Q
 \le\delta\log\frac{(d-1)(1-\alpha)}{\alpha}
       -\log(1-\alpha)
 =g_d(\delta).
\]
Exchanging $P,Q$ gives the absolute value. For sharpness, take $Y$ deterministic, let $Q_X$ be a point mass,
and, with $t=\min\{\delta,1-d^{-1}\}$, let $P_X$ assign mass $1-t$ to
that point and $t/(d-1)$ to every other point.  Then $T(P,Q)=t$ and
$H(X|Y)_P-H(X|Y)_Q=g_d(\delta)$.
\end{proof}

\section{Winter's proof}
\label{app:winter}
 
With $t=T(\rho_{AB},\sigma_{AB})$, Winter's strengthened
Alicki--Fannes bound reads
\cite{Winter2016}
\[
 \abs{H(A|B)_\rho-H(A|B)_\sigma}
 \le2t\log d+(1+t)h_2\!\left(\frac{t}{1+t}\right).
\]
Winter's proof writes $\rho_{AB}-\sigma_{AB}
=t(\Delta_{+,AB}-\Delta_{-,AB})$ via the Jordan decomposition and
treats the two states symmetrically, i.e., each is mixed toward the common
state
\[
 \omega_{AB}
 =\frac{\sigma_{AB}+t\Delta_{+,AB}}{1+t}
 =\frac{\rho_{AB}+t\Delta_{-,AB}}{1+t}.
\]
The bound then follows from approximate affinity of the
conditional entropy together with $-\log d\le H(A|B)\le\log d$. The conditional-entropy range gives the term $2t\log d$.
Concavity for one decomposition and the upper mixing bound for the other
give a single penalty $h_2(t/(1+t))$; multiplying by $1+t$ produces
$(1+t)h_2(t/(1+t))$.
 
The present proof breaks this symmetry and builds the comparison
state from $\sigma_{AB}$ alone:
\[
 \widehat{\sigma}_{AB}
 =\frac{d\,\id_A\otimes\sigma_B-\sigma_{AB}}{d^2-1},
 \qquad
 \tau_{\alpha,AB}
 =(1-\alpha)\sigma_{AB}+\alpha\widehat{\sigma}_{AB},
 \qquad
 \alpha=\min\{t,1-d^{-2}\}.
\]
Thus $\rho_{AB}$ enters only through the trace-distance
hypothesis, and by construction $\tau_{\alpha,B}=\sigma_B$.  The
expansion at the start of the proof of Theorem~\ref{thm:main} then
converts the nonlinear remainder into a difference of relative
entropies whose sign is controlled by data processing, while the
two-sided comparison of $\tau_{\alpha,AB}$ with $\id_A\otimes\sigma_B$
gives a spectral-width bound of
$\log((d^2-1)(1-\alpha)/\alpha)$ for $\alpha>0$. The factor $d^2-1$ is the
trace of $d\,\id_A\otimes\sigma_B-\sigma_{AB}$; the construction does not
require $\widehat\sigma_{AB}$ to be orthogonal to $\sigma_{AB}$.
On the increasing branch, the choice $\alpha=t$
yields exactly the binary-entropy term $h_2(t)$; at the endpoint $\alpha=1-d^{-2}$, the
operator interval containing $G_{AB}$ shrinks to a point, which
produces the plateau.

\end{document}